\documentclass[11pt]{article}
\usepackage{amsthm}
\usepackage{amssymb}
\usepackage{amsfonts}
\usepackage{amsmath}
\usepackage{mathrsfs}
\usepackage{mathtools}
\usepackage{lmodern}
\usepackage{graphicx}
\usepackage{listings}

\usepackage{algorithm}
\usepackage{algorithmic}
\usepackage{color,xcolor}
\usepackage[inline]{enumitem}
\usepackage{multirow}
\usepackage{diagbox}
\usepackage[colorlinks]{hyperref}
\usepackage{cleveref}
\usepackage{makecell}
\usepackage{verbatim}
\usepackage{geometry}
\usepackage{setspace}
\usepackage{caption}
\usepackage{tikz}
\usepackage{tikz-3dplot}
\usepackage{pgfplots}
\pgfplotsset{compat=1.17}
\usetikzlibrary{arrows.meta}
\usepackage{bbding}
\usepackage{threeparttable}

\allowdisplaybreaks[3]

\crefname{construction}{Construction}{Constructions}

\newcommand{\R}{\mathbb{R}}
\newcommand{\A}{\mathcal{A}}

\newcommand{\norm}[1]{\left\lVert #1\right\rVert}
\newcommand{\inn}[2]{\left\langle {#1},{#2}\right\rangle}

\newcommand{\poly}{\mathrm{poly}}
\newcommand{\rank}{\mathrm{rank}}
\newcommand{\grad}{\mathrm{grad}}
\newcommand{\diff}{\mathrm{d}}

\theoremstyle{plain}
\newtheorem{theorem}{Theorem}[section]
\newtheorem{lemma}{Lemma}[section]
\newtheorem{proposition}{Proposition}[section]
\newtheorem{corollary}{Corollary}[section]

\theoremstyle{definition}
\newtheorem{definition}{Definition}[section]
\newtheorem{example}{Example}[section]

\theoremstyle{remark}
\newtheorem{remark}{Remark}[section]

\title{On Convergence Lemma and Convergence Stability for Piecewise Analytic Functions}
\author{Xiaotie Deng\\ xiaotie@pku.edu.cn \\ Peking University
\and 
Hanyu Li\\ lhydave@pku.edu.cn \\ Peking University
\and 
Ningyuan Li \\ liningyuan@pku.edu.cn \\ Peking University
}

\begin{document}
\maketitle
\begin{abstract}
In this work, a convergence lemma for function $f$ being finite compositions of analytic mappings and the maximum operator is proved. The lemma shows that the set of $\delta$-stationary points near an isolated local minimum point $x^*$ is shrinking to $x^*$ as $\delta\to 0$. It is a natural extension of the version for strongly convex $C^1$ functions. However, the correctness of the lemma is subtle. Analytic mappings are necessary for the lemma in the sense that replacing it with differentiable or $C^\infty$ mappings makes the lemma false. The proof is based on stratification theorems of semi-analytic sets by {\L}ojasiewicz. An extension of this proof presents a geometric characterization of the set of stationary points of $f$. Finally, a notion of stability on stationary points, called convergence stability, is proposed. It asks, under small numerical errors, whether a reasonable convergent optimization method started near a stationary point should eventually converge to the same stationary point. The concept of convergence stability becomes nontrivial qualitatively only when the objective function is both nonsmooth and nonconvex. Via the convergence lemma, an intuitive equivalent condition for convergence stability of $f$ is proved. These results together provide a new geometric perspective to study the problem of ``where-to-converge'' in nonsmooth nonconvex optimization.
\end{abstract}
\section{Introduction}
We consider the following optimization problem
\[\mathrm{Minimize}\quad f(x),\quad x\in\Omega,\]
where $f$ is finite compositions of the maximum operator and analytic mappings, and $\Omega$ is a subset of Euclidean space $\R^n$. Function $f$ is usually nonconvex everywhere and non-differentiable on a zero-measure set. It satisfies the assumptions in many works on general nonsmooth nonconvex optimization \cite{bagirov2014introduction,Burke2002GradientSampling,BurkeRobustGradientSampling2005,KiwielBundle,DBLP:journals/siamjo/Kiwiel07}. Although $f$ seems more specific and regular than functions considered in the literature, such an objective function involves a majority of popular loss functions in deep learning. Indeed, most deep neural networks are comprised of ingredients being either an analytic mapping or a maximum operator, such as convolutions \cite{Krizhevsky17AlexNet,Lecun98CNN}, residual networks \cite{He16Residual}, normalization operators \cite{Ioffe15batchnorm}, pooling operators \cite{DBLP:journals/corr/abs-2009-07485} and attentions \cite{Vaswani17Attention}. 

A series of widely used first-order methods have been developed for problems of this kind, including gradient descent methods \cite{Curry1944Steepest}, standard subgradient methods \cite{Shor85subgradient}, and gradient sampling methods \cite{Burke2002GradientSampling}. Most algorithms aim at local minimization, i.e., to find a \emph{stationary point} or a \emph{$\delta$-stationary point}. In smooth settings, a stationary point $x$ is usually defined to be the one with zero gradient \cite{boyd2004convex}. For nonsmooth functions, however, differentiable points are usually with nonzero gradients. Even worse, local minima are usually non-differentiable. The standard definitions of stationary points for nonsmooth functions are given by Rockafellar \cite{Rockafellar+2015} and Clarke \cite{clarke1973necessary, clarke1990optimization}, based on subdifferentials. Naturally, function $f$ fits these definitions. 

One main concern about these algorithms is their convergence property: Does the iteratively generated point sequence converge to a stationary point? For strongly convex $C^1$ functions with minimum points, the convergence theory is well-developed. There, the main property to guarantee a convergence is that the decreasing of the function value is equivalent to the decreasing of the gradient norm \cite{boyd2004convex}. Alternatively, $\delta$-stationary points converge to the unique stationary point as $\delta\to 0$. However, this is not always true for nonconvex functions: Stationary points are usually not unique. Even when we restrict to a neighborhood of an isolated stationary point, this claim can fail for nonsmooth functions since the subdifferential may not be continuous. Convergence analyses for nonsmooth nonconvex functions have been developed for specific iterative methods, e.g., the gradient sampling methods \cite{Burke2002GradientSampling,BurkeRobustGradientSampling2005,DBLP:journals/siamjo/HosseiniU17,DBLP:journals/siamjo/Kiwiel07}. Yet no systematic and general approaches are devised.

In this work, we prove the following convergence lemma, extending the basic convergence result in strongly convex $C^1$ cases.

\begin{theorem}[Convergence lemma]\label{thm:analytic-abs-delta-SP-converge}
Suppose $f(x)$ is finite compositions of analytic mappings and the maximum operator defined on a cube $\Omega$.\footnote{$\Omega$ is only required to be a locally compact semi-analytic set. To avoid technical details, we here replace it with a cube.} If stationary point $x^*$ is an isolated local minimum point, then there exists a radius $r_1>0$ such that for any radius $\varepsilon>0$, there exists $\delta_0>0$, such that every $\delta$-stationary point in ball $B(x^*,r_1)\cap\Omega$ is in ball $B(x^*,\varepsilon)\cap\Omega$ whenever $\delta<\delta_0$.
\end{theorem}
Less formally, \Cref{thm:analytic-abs-delta-SP-converge} means that the set of $\delta$-stationary points near any isolated local minimum point is guaranteed to shrink when $\delta$ is getting smaller. Hence in the meaning of gradient norms, function $f$ locally behaves like a strongly convex $C^1$ function at an isolated local minimum point.

\Cref{thm:analytic-abs-delta-SP-converge} heavily relies on the analyticity property. For functions being finite compositions of the maximum operator and differentiable (or even $C^\infty$) mappings\footnote{It means the mapping is differentiable or infinitely differentiable.}, counterexamples exist (\Cref{ex:nonex-for-converge-stable-C1} and \Cref{ex:nonex-for-converge-stable}). Even when we in addition require $x^*$ to be an isolated stationary point, \Cref{thm:analytic-abs-delta-SP-converge} still fails to hold in these examples. Such subtlety suggests that analyticity is to some extent necessary for this intuitive convergence lemma. 

The main idea of the proof is to partition the domain into \emph{finitely many} regions so that \emph{any} sequence of $\delta$-stationary points with $\delta\to 0$ has a subsequence lying in some particular well-conditioned region. As $\delta\to 0$, the limit behavior of $\delta$-stationary points is then well confined. To utilize analyticity, we carefully partition the domain into regions being both \emph{semi-analytic sets} and \emph{analytic manifolds}. The technique we use is \emph{stratification theorems}, proposed by Whitney \cite{WhitneyTangents} for complex analytic sets and extended to real semi-analytic sets by {\L}ojasiewicz \cite{lojasiewicz1965ensembles, lojasiewicz1995semi}.

Inspired by Bolte, Daniilidis and Lewis's work \cite{DBLP:journals/siamjo/BolteDL07}, we further extend our proof and characterize the geometric properties of the set of \emph{all} stationary points of function $f$. Specifically, we have the following theorem.
\begin{theorem}[Geometric characterizations on stationary points]\label{prop:geometric-char-of-sp}
Stationary points of $f$ form a finite union of connected subsets $A_i$, each sharing \textbf{all} the following properties:
\begin{enumerate}[label=\roman{*}.]
    \item $A_i$ is a sub-analytic set as well as an analytic submanifold (possibly a single point).
    \item Each point in $A_i$ is a stationary point.
    \item $f$ is a constant on $\bar{A}_i$, the closure of $A_i$.
    \item If $x^*$ is a non-isolated local minimum point, then there exists a neighborhood $U$ of $x^*$, such that for every $A_i$ and every $x\in A_i\cap U$, $x$ is also a non-isolated local minimum point.
\end{enumerate}
\end{theorem}

A direct corollary from \Cref{prop:geometric-char-of-sp} is that there are only \emph{finite} possible function values taken by \emph{all} stationary points, or local minima.

Apart from such geometric characterizations, \Cref{thm:analytic-abs-delta-SP-converge} also brings in a natural stability notion on stationary points: \emph{convergence stability}. It focuses on whether a reasonable iterative method started near a stationary point should be expected to stay near and eventually converge to this stationary point. Such a stability concept tolerates small numerical errors produced in the optimization process. For most optimization methods dealing with continuous problems, numerical errors during the optimization process may lead to gaps between the theoretical analysis and the observed behavior. When a stationary point is convergence-stable, it indicates that the theoretical results are insensitive to implementation issues. Such considerations are nontrivial qualitatively only when the objective function is both nonsmooth and nonconvex. Due to analyticity, \Cref{thm:analytic-abs-delta-SP-converge} yields the following theorem for function $f$, presenting a full geometric understanding of the convergence stability.

\begin{theorem}[Equivalent condition of convergence stability]\label{thm:analytic-max-is-convergence-stable}
Suppose $f(x)$ is finite compositions of analytic mappings and the maximum operator defined on a cube $\Omega$.\footnote{Again, $\Omega$ is only required to be a locally compact semi-analytic set.} Let $x^*$ be a stationary point of $f$. $x^*$ is convergence-stable if and only if $x^*$ is an isolated local minimum point.
\end{theorem}

\Cref{thm:analytic-max-is-convergence-stable} is intuitive, but by no means trivial. Similar to \Cref{thm:analytic-abs-delta-SP-converge}, analog counterexamples of \Cref{thm:analytic-max-is-convergence-stable} exist (\Cref{ex:nonex-for-converge-stable-C1} and \Cref{ex:nonex-for-converge-stable}) when the analytic mappings are replaced with differentiable or $C^\infty$ ones. Thus techniques related to analytic mappings are crucial for the proof.

\Cref{thm:analytic-abs-delta-SP-converge} and \Cref{thm:analytic-max-is-convergence-stable} together study the problem of ``where-to-converge'' in nonsmooth nonconvex optimization. These results provide a new geometric understandings to the dynamics of optimization processes. As an application in computational game theory, \Cref{thm:analytic-max-is-convergence-stable} provides an affirmative answer to the existence of stable instances for Tsaknakis-Spirakis descent methods for approximate Nash equilibria \cite{DBLP:conf/sagt/ChenDHLL21,DBLP:journals/corr/abs-2204-11525,DBLP:journals/im/TsaknakisS08,DBLP:conf/wine/TsaknakisSK08}.

This paper is organized as follows. In \Cref{sec:pre}, we introduce definitions and notations. Then in \Cref{sec:necessity-of-analytic} we show by two examples the necessity of analyticity in the convergence lemma. Next, we present the proof of the convergence lemma in \Cref{sec:proof-convergence-lemma}. Extending this proof, we prove \Cref{prop:geometric-char-of-sp} in \Cref{sec:locus-of-SP}. In \Cref{sec:convergence-stability}, we introduce the notion of convergence stability and then prove \Cref{thm:analytic-max-is-convergence-stable}. Finally, we conclude our results and present the next immediate open issues in \Cref{sec:discussion}.
\section{Preliminaries}\label{sec:pre}
We first introduce frequently used definitions and notations.

\subsection*{Sets and Euclidean spaces} The $n$-dimensional Euclidean space is denoted by $\R^n$. The closure of a set $A\subseteq\R^n$ is denoted by $\bar{A}$. The Euclidean inner product is $\inn{\cdot}{\cdot}$. The Euclidean norm in $\R^n$ is $\norm{\cdot}$. The open ball centered at $x$ with radius $r$,  $\{x'\in\R^n:\norm{x-x'}<r\}$, is denoted by $B(x,r)$. A closed ball can be expressed by $\bar{B}(x,r)$. Denote by $0_k$ a $k$-dimension column vector with all entries equal to $0$.

\subsection*{Mappings and functions} Let $F$ be a mapping. If $F$ maps to $\R$, we call $F$ a function. Below we suppose $F:\Omega\to\R^n$ is a mapping. The following notations and definitions apply directly to functions as well. 

The image of $S\subseteq \Omega$ is $F(S):=\{F(x):x\in S\}$. The pre-image of $S\subseteq\R$ is $F^{-1}(S):=\{x\in \Omega:F(x)\in S\}$. The restriction of $F$ on $S\subset\Omega$ is $F|_S$. The composition of two mappings $F$ and $G$ is $(F\circ G)(x):=F(G(x))$. If $F$ is univariate, denote its derivative at $x$ by $F'(x)$ or $\diff F(x)/\diff x$. Suppose $F$ is multi-variate. The directional derivative of $F$ in direction $s$ at $x$ is $\partial F(x)/\partial s$. If $F$ is a function, the gradient of $F$ at $x$ is $\grad F(x)$. The Jacobi matrix of $F$ at $x$ is $J_F(x)$ and the rank of this matrix is $\rank_x(F)$. $F\in C^n(\Omega)$ if it is continuously differentiable of the $n$th order on $\Omega$. $F\in C^\infty(\Omega)$ if $F\in C^n(\Omega)$ for all positive integer $n$. We say $F$ is analytic if for every $x\in \Omega$, $F$ is equal to its Taylor expansion at $x$ in some neighborhood of $x$.\footnote{Analytic mappings are $C^\infty$ mappings. The reverse is not true. One counterexample is Cauchy function $F(x)=\exp(x^{-2})$, which is $C^\infty$ but not analytic at $x=0$.}

We say a function $f(x)$ defined on $\Omega\subseteq \R^n$ is \emph{finite compositions of analytic mappings and the maximum operator}, if 
\[f(x)=f_{m}(f_{m-1}(\cdots f_{2}(f_{1}(x)))),\]
where $f_i:\Omega_{i-1}\to\Omega_{i}$, $\Omega_{i}\subseteq \R^{d_{i}}$, $d_{0}=n,d_{m}=1,\Omega_{0}=\Omega$, and $f_i(z_1,z_2,\cdots,z_{d_{i-1}})$ is either an analytic mapping on $\Omega_{i-1}$ , or in the form of 
\[f_i(z_1,z_2,\cdots,z_{d_{i-1}})=\left(\max_{j\in S_{i,1}} z_j,\max_{j\in S_{i,2}} z_j,\cdots,\max_{j\in S_{i,d_{i}}} z_j\right),\] 
where $S_{i,k}\subseteq \{1,\cdots,d_{i-1}\}$ for $k=1,\cdots,d_{i}$.

It sometimes simplifies notations to consider \emph{Dini directional derivatives} \cite{Demyanov2010}.
\begin{definition}
Suppose $x,x'\in \Omega$. The \emph{Dini directional derivative} in direction $x'-x$ at point $x$, denoted by $Df(x,x')$, is defined as the following limit (if the limit exists)
\[\lim_{\alpha\downarrow 0}\frac{f(x+\alpha (x'-x))-f(x)}{\alpha}.\]
\end{definition}
Note that standard directional derivatives is equivalently defined as the following limit:
\[\frac{\partial f(x)}{\partial s}:=\lim_{\alpha\downarrow 0}\frac{f(x+\alpha s)-f(x)}{\norm{\alpha s}}.\]
The relation between the two limits is direct.
\begin{lemma}\label{lemma:two-derivative}
Let $s=x'-x$. Then $\frac{\partial f(x)}{\partial s}$ exists if and only if $Df(x,x')$ exists. Moreover,
\[Df(x,x')=\frac{\partial f(x)}{\partial s}\norm{s}.\]
\end{lemma}

We note that function $f$ considered in this work has (Dini) directional derivative in every direction at every point since the maximum operator is directional differentiable.

We present two subdifferential concepts developed by Rockafellar \cite{Rockafellar+2015} and Clarke \cite{clarke1973necessary, clarke1990optimization}. The standard subdifferential at $x$ is defined to be
\[\partial f(x):=\{d\in\R^n:f(x')-f(x)\geq\inn{d}{x'-x}\text{ for all }x'\in\Omega\}.\]

The Clarke subdifferential $\bar{\partial} f(x)$ at $x$ is defined by the convex hull of the limits of gradients of $f$ on sequences converging to $x$, i.e., 
\[\bar{\partial} f(x):=\mathrm{conv}\left\{\lim_{i\to\infty} \grad f(y_i):\{y_i\}_{i\geq 0}\to x\text{, where }f\text{ is differentiable at every }y_i\right\}.\]

Again, function $f$ admits these subdifferentials at every point in $\Omega$.

\subsection*{Terminologies in Optimization}
Now we involve some concepts in optimization theory. From the optimization perspective, an optimization algorithm is designed to search for a minimum point (or local minimum point) of an objective function. In most settings, the optimization procedure is equivalent to finding a \emph{stationary point}. Due to the regularity of function $f$, such a concept can be defined as follows.
\begin{definition}
$x\in \Omega$ is a \emph{stationary point} if $\partial f(x)/\partial s\geq 0$ for every direction $s$. Or equivalently, $Df(x,x')\geq 0$ for every $x'\in\Omega$.
\end{definition}

Due to precision restrictions, it is almost impossible for most optimization algorithms to find an exact stationary point. Thus \emph{approximate stationary points} are proposed.
\begin{definition}
Let $\delta>0$. $x\in \Omega$ is a $\delta$-\emph{stationary point} if $\partial f(x)/\partial s\geq -\delta$ for every direction $s$.
\end{definition}

We remark that both definitions degenerate to the standard definitions (e.g., in \cite{boyd2004convex}) when $f$ is differentiable. Indeed, suppose $f$ is differentiable at $x\in\Omega$. Then $\grad f(x)=0\iff\partial f(x)/\partial s=0$ for every direction $s$, so the definition of stationary points coincides. On the other hand, $\norm{\grad f(x)}=\sup_{\norm{s}=1}\inn{\grad f(x)}{s}=\sup_s|\partial f(x)/\partial s|$, so the definition of $\delta$-stationary points coincides. Plus, stationary points defined by stand subdifferentials \cite{Rockafellar+2015} are equivalent to the above ones. Indeed, $0\in\partial f(x)\iff Df(x,x')\geq 0$. Similar facts hold for the approximate stationary points. In addition, our notions on (approximate) stationary points are a subset of Clarke (approximate) stationary points \cite{clarke1990optimization}.

Finally, we adopt the convention to present our results in the form of theorems, lemmas and corollaries, and results in the existing literature in the form of propositions.
\section{Necessity of Analyticity}\label{sec:necessity-of-analytic}
In this section, we show by two counterexamples that the analyticity assumption in \Cref{thm:analytic-abs-delta-SP-converge} is in some way necessary.

The first example shows that replacing analytic mappings by differentiable mappings is not enough for \Cref{thm:analytic-abs-delta-SP-converge}.
\begin{example}\label{ex:nonex-for-converge-stable-C1}
Let 
\begin{align*}
g(x)&:=\int_0^x 2u\sin(1/u) du,\\
h(x)&:=\begin{cases}x^2\sin(1/x),&x\neq 0,\\
0, &x=0,\end{cases}\\
f(x)&:=\min\{0,h(x)\}-g(x)+|x|.
\end{align*}

We show that $x^*=0$ is the unique stationary point as well as the unique local minimum point in $B\left(0,\frac1\pi\right)$. In the meantime, however, \Cref{thm:analytic-abs-delta-SP-converge} fails for $x^*$. 

Note that $g(x)$ and $h(x)$ are differentiable on $\R$:
\begin{align*}
    g'(x)&=\begin{cases}2x\sin(1/x),&x\neq 0,\\
    0,&x=0.
    \end{cases}\\
    h'(x)&=\begin{cases}2x\sin(1/x)-\cos(1/x),&x\neq 0,\\
    0,&x=0.
    \end{cases}
\end{align*} 

Therefore $f(x)$ can be expressed as finite compositions of differentiable mappings and the maximum function. However, we note that $f(x)$ has infinitely many non-differentiable points in any neighborhood of $x=0$.

Let $f'_{-}(x)=\lim_{x'\uparrow x}\frac{f(x')-f(x)}{x'-x}$ and $f'_{+}(x)=\lim_{x'\downarrow x}\frac{f(x')-f(x)}{x'-x}$ denote the left and right derivative of $f$ at $x$, respectively. By definition, $x$ is a stationary point of $f$ if and only if $f'_{-}(x)\leq0$ and $f'_{+}(x)\geq 0$.

First, we show that
\[\forall x\in(0,1), f'_{-}(x)>0.\] 
Observe that for $x>0$, $h(x)=0$ if and only if $x=\frac1{k\pi}$ for some integer $k>0$.

For $x\in \left(\frac1{(2k+1)\pi},\frac1{(2k-1)\pi}\right]$, $k=1,2,\cdots$, we discuss the left and right derivative of $f$ at $x$.

If $x\in\left(\frac1{2k\pi},\frac1{(2k-1)\pi}\right)$, $f(x)=h(x)-g(x)+x$, so $f^{\prime}_{-}(x)=f^{\prime}_{+}(x)=f'(x)=h'(x)-g'(x)+1=1-\cos(1/x)>0$.

If $x\in\left(\frac1{(2k+1)\pi},\frac1{2k\pi}\right)$, $f(x)=0-g(x)+x$, so $f^{\prime}_{-}(x)=f^{\prime}_{+}(x)=f'(x)=-g'(x)+1=1-2x\sin(1/x)>1-2x>0$.

If $x=\frac1{(2k-1)\pi}$, $f^{\prime}_{-}(x)=1-\cos(1/x)=2,f^{\prime}_{+}(x)=1-2x\sin(1/x)=1$.

For $x=\frac1{2k\pi}$, $f^{\prime}_{-}(x)=1-2x\sin(1/x)=1,f^{\prime}_{+}(x)=1-\cos(1/x)=0$.

Therefore, for any $x\in(0,\frac1\pi)$, $f^{\prime}_{-}(x)>0$, so there is no stationary point in $(0,\frac1\pi)$.
Similarly we can prove that for any $x\in\left(-\frac1\pi,0\right)$, $f^{\prime}_{+}(x)<0$, so $x$ cannot be a stationary point. In addition, one can check that $f^{\prime}_{+}(0)=1$ and $f^{\prime}_{-}(0)=-1$, so $x^*=0$ is a stationary point. Moreover, $x^*=0$ is also an isolated local minimum point.

However, we show that \Cref{thm:analytic-abs-delta-SP-converge} does not hold for $x^*=0$. Observe that for $k=1,2,\cdots$, and $t\in(0,\pi)$, \[f'\left(\frac1{2k\pi-t}\right)=1-\cos\left(2k\pi-t\right)=1-\cos\left(t\right)\leq\frac{1}{2}t^2.\] 
For any $r_0>0$, take $k=\left\lceil\frac1{2\pi r_0}\right\rceil+2$, and $\epsilon=\frac1{2k\pi+1}$. For any $\delta>0$, let $t=\min\{\sqrt{2\delta},1\}$. Set $x_0=\frac1{2k\pi-t}\in B(x^*,r_0)$. Then $x_0$ is a $\delta$-stationary point. However, $x_0=\frac1{2k\pi-t}\notin B(x^*,\epsilon)$. This violates \Cref{thm:analytic-abs-delta-SP-converge}.
\end{example}

We further construct the following similar counterexample which is finite compositions of $C^\infty$ mappings and the maximum operator.
\begin{example}\label{ex:nonex-for-converge-stable}
Let 
\begin{align*}
g(x)&:=\begin{cases}\displaystyle\int_0^x \left(u^{-2}e^{-1/\left|u\right|}\sin\left(1/u\right)-u^{-2}e^{-1/\left|u\right|}\right) du,&x\neq 0,\\
0, &x=0,\end{cases}\\
h(x)&:=\begin{cases}e^{-1/\left|x\right|}\sin(1/x),&x\neq 0,\\
0, &x=0,\end{cases}\\
f(x)&:=\min\{0,h(|x|)\}-g(|x|).
\end{align*}
One can verify the following facts:
\begin{enumerate}
    \item For any $x>0$,
    \begin{align*}
    g'(x)&=x^{-2}e^{-1/x}\sin\left(1/x\right)-x^{-2}e^{-1/x},\\
    h'(x)&=x^{-2}e^{-1/x}\sin\left(1/x\right)-x^{-2}e^{-1/x}\cos\left(1/x\right).
    \end{align*}
    \item For both $g(x)$ and $h(x)$, the differential of every order at $x=0$ is $0$.
    \item $g(x)$ and $h(x)$ belong to $C^\infty(\R)$.
    \item $f(x)$ can be expressed as finite compositions of $C^\infty$ mappings and the maximum function.
    \item $f'_{-}(x)>0$ for $x\in(0,1)$, so $f(x)$ is strictly increasing on $(0,1)$. These hold symmetrically on $(-1,0)$. So $x^*=0$ is the only stationary point in $B(0,1)$, and is the unique local minimum point.
    \item However, \Cref{thm:analytic-abs-delta-SP-converge} does not hold for $x^*=0$.
\end{enumerate}
The proof is quite similar to the previous example, so we put the detailed proof in \Cref{app:nonexample-convergence-stable}.
\end{example}

\section{Proof of the Convergence Lemma}\label{sec:proof-convergence-lemma}
In this section, we prove the convergence lemma, i.e., \Cref{thm:analytic-abs-delta-SP-converge}. The proof of \Cref{thm:analytic-abs-delta-SP-converge} can be divided into three steps. First, we rewrite $f$ into finite compositions of analytic mappings and the absolute-value operator. Then, we partition the domain $\Omega$ into so-called \emph{semi-analytic sets} and show tight connections between our problem and properties of semi-analytic sets. Finally, we utilize these connections to prove the lemma by contradiction. 

We will prove \Cref{thm:analytic-abs-delta-SP-converge} for compact semi-analytic domain $\Omega$ after introducing semi-analytic sets.

\subsection{From the Maximum Operator to the Absolute-Value Operator}
To make the proof more concise, we use the absolute-value operator $|x|$ to replace the maximum operator. A standard method to replace maximum operator with absolute-value operator is
\[\max\{a,b\}=\frac{|a-b|+a+b}{2}.\]
Thus finite compositions of analytic mappings and the maximum operator are equivalent to finite compositions of analytic mappings and the absolute-value operator. To make notations simple, we will consider the \emph{canonical form} below.

Consider $f:\Omega\to \R$, where $\Omega\subseteq\R^n$ is compact. $f$ is defined by the following evaluation process:
\begin{align*}
    z_1:&=f_1(x),\\
    z_2:&=f_2(x,|z_1|),\\
    &\cdots\\
    z_i:&=f_i(x,|z_1|,|z_2|,\cdots,|z_{i-1}|),\\
    &\cdots\\
    z_m:&=f_m(x,|z_1|,|z_2|,\cdots,|z_m|),\\
    f(x):&=z_m,
\end{align*}
where $f_i:\R^{n+i-1}\to\R$. For convenience, we use $z_1(x),z_2(x),\cdots,z_m(x)$ to denote the value of these intermediate variables as functions of $x$. We further require that for every $i=1,\cdots,m+1$, $f_i$ is analytic on the compact set
\[\Omega_{i}:=\{(x,|z_1(x)|,|z_2(x)|,\cdots,|z_{i-1}(x)|):x\in\Omega\}\subseteq\R^{n+i-1}.\] 
It is clear that every function in the form of finite compositions of analytic mappings and the absolute-value operator can be expressed in such canonical form.

\subsection{Semi-Analytic Geometry}
Before proposing the approach to partitioning $\Omega$, we introduce \emph{semi-analytic sets} and their geometric properties. We will use them for later proofs. We follow the standard terminologies from \cite{Semianalytic-and-subanalytic-sets} and \cite{lojasiewicz1995semi}.

A set $L\subseteq\R^n$ is called \emph{semi-analytic} if for every point $\alpha\in\R^n$, there is a neighborhood $U_\alpha$ of $\alpha$ such that $L\cap U_\alpha$ is a finite union of sets of the form $\{x\in\R^n:f_i(x)=0,g_j(x)> 0, i=1,\cdots,i_0,j=1,\cdots,j_0\}$, where $f_i$ and $g_j$ are analytic functions on $U_\alpha$. An equivalent statement of semi-analyticity is that $L$ is locally\footnote{A property $P$ is said to \emph{locally} hold for set $L\subseteq\R^n$ if for every point $\alpha\in\R^n$, there is a neighborhood $U_\alpha$ of $\alpha$ such that $P$ holds for $L\cap U_\alpha$. } a finite union of sets of that form. 

We now instead suppose the domain $\Omega$ is a compact semi-analytic set and prove the theorem for functions on this kind of domains. We have the following properties of semi-analytic sets.

\begin{proposition}[\cite{lojasiewicz1965ensembles}]
\label{prop:semi-analytic-set-properties}
Let $L\subseteq\R^n$ be a semi-analytic set. The following statements hold: 
\begin{itemize}
    \item The closure $\bar{L}$ is a semi-analytic set.
    \item The union and the intersection of a finite collection of semi-analytic sets is semi-analytic.
    \item $L$ is a locally finite union of connected semi-analytic sets ($L$'s connected components).
    \item If $L$ is connected, then any two points, $a_1$ and $a_2$, in $L$ can be joined by a semi-analytic curve. This means that there exists some embedding $\varphi: [0, 1]\to L$ whose image is a semi-analytic set in $\R^n$, with $\varphi(0) = a_0$ and $\varphi(1) = a_2$. A semi-analytic curve is analytic everywhere with the exception of a finite number of points.
\end{itemize}
\end{proposition}

By Borel-Heine finite cover lemma, statement ``locally finite union'' can be improved to ``finite union on compact $\Omega$''. Another observation is that a semi-analytic set restricted on semi-analytic domain $\Omega$ (i.e., intersecting with $\Omega$) is still semi-analytic. We will use these two facts to improve every local property of semi-analytic sets to a global property on $\Omega$.

We have the following property.
\begin{proposition}[In the proof of Lemma 3.4 in \cite{Semianalytic-and-subanalytic-sets}]
\label{prop:semi-analytic-set-rank}
Let $L\subseteq\R^n$ be a semi-analytic set. Let $\pi_{n,m}:\R^n\to\R^m$ be the projection mapping 
\[(x_1,x_2,\cdots,x_n)\mapsto (x_1,x_2,\cdots,x_m),\]
where $m<n$. If $L$ is a connected analytic manifold, then for every $k\leq m$, $\{x\in L:\rank_{x}(\pi_{n,m}\vert_L)\leq k\}$ is a semi-analytic set.
\end{proposition}

To present further properties, we need \emph{stratification theorems} motivated by Whitney \cite{WhitneyTangents}. A stratification tries to divide a manifold into a finite set of submanifolds based on some given sets. In semi-analytic situations, we state such stratifications as follows. Given an open cube $Q=\{(x_1,\cdots,x_n)\in\R^n:|x_i|<\delta\}$, a \emph{normal stratification} in $\R^n$ is a finite decomposition $\mathcal{I}=\{\Gamma_1,\cdots,\Gamma_g\}$ of $Q$ into disjoint union of semi-analytic sets $\Gamma_j$'s, called \emph{strata}. A stratum $\Gamma_j\in\mathcal{I}$ in this
stratification is a subset of $Q$ being both an analytic submanifold and a semi-analytic set. 

Given an analytic $n$-manifold $M$ and a point $x\in M$, take a local chart $(U,\xi)$ of $M$ at $x$, where $U$ is the homeomorphic neighborhood of $x$ and $\xi$ is the analytic coordinate transformation from $U$ to $\R^n$ such that $\xi(x)=0_n$. A \emph{normal stratification} of $M$ at $x$ is defined to be an image of a stratification in $\R^n$ under the inverse map $\xi^{-1}$. {\L}ojasiewicz proved the following theorem.
\begin{proposition}[Theorem on normal stratification, \cite{lojasiewicz1965ensembles}]
For arbitrarily given semi-analytic sets $E_1,\cdots, E_k$ in $M$ and a point $x\in M$, there exists a normal stratification of $M$ in an arbitrarily small neighbourhood $U$ of $x$ such that each of the sets $E_i\cap U$ is a union of some strata of this stratification.
\end{proposition}

In fact, it can be required that the strata satisfy more conditions. Let $M$ be a manifold in $\R^n$. Take $x\in M$. Denote tangent space\footnote{A tangent space of set $S$ at point $x\in S$ is the set of tangent vectors of $S$ at $x$. If $S$ is a manifold, any tangent space of $S$ is also a vector space.} at $x$ by $ T_x M$. Let $N$ be another manifold in $\R^n$. Define the distance between two tangent spaces contained in $\R^n$ as follows.
\[\delta( T_x M,  T_y N):=\max_{v\in T_x M,\norm{v}=1}\min_{u\in T_y N,\norm{u}=1}\norm{u-v}.\]
Note that the minimum and the maximum can be attained since the set of unit vectors in a tangent space is compact and the distance function is continuous.

A \emph{semi-analytic stratification} of an analytic manifold $M$ is a locally finite decomposition $\mathcal{I}$ of this manifold into semi-analytic strata, fulfilling the additional condition: For every stratum $\Gamma\in\mathcal{I}$, its boundary $\partial \Gamma = \bar{\Gamma}\setminus\Gamma$ is a finite union of strata of family $\mathcal{I}$ (having dimensions less than $\dim\Gamma$).

Let us consider triples $(\Lambda,\Gamma,a)$, where $\Lambda, \Gamma$ are semi-analytic strata of an analytic manifold $M$, such that $\Lambda\subseteq\partial\Gamma$ and $a\in\Lambda$. A well-studied condition proposed by Whitney \cite{WhitneyTangents} is 
\[\delta( T_a\Lambda, T_x\Gamma)\to 0\quad\text{when}\quad\Gamma\ni x\to a.\]
This condition is called \emph{Whitney's condition (a)}. We have the following stratification theorem.
\begin{proposition}[Theorem on semi-analytic stratification, \cite{lojasiewicz1995semi,Stratifications-Whitney}]
\label{thm:semi-analytic-stratification}
For every locally finite family of semi-analytic subsets of an analytic manifold $M$, there exists a semi-analytic stratification consistent with this family and such that for every pair $\Lambda, \Gamma$ of strata of this stratification, such that $\Lambda\subseteq\partial\Gamma$, the Whitney's condition (a) is fulfilled at every point of the stratum $\Lambda$. 
\end{proposition}
A semi-analytic set is called a \emph{smooth semi-analytic set} if it is also an analytic manifold. We have a direct corollary from \Cref{thm:semi-analytic-stratification} and \Cref{prop:semi-analytic-set-properties}.
\begin{lemma}\label{prop:semi-analytic-Whitney}
Let $L\subseteq\R^n$ be a semi-analytic set. Suppose the closure $\bar{L}$ is compact. Then $L$ is the disjoint union of a finite collection $\{A_k\}$ of connected smooth semi-analytic subsets. Moreover, $A_k\cap\partial A_l\neq\varnothing$ implies that $A_k\subseteq\partial A_l$, and that $A_k$ and $A_l$ satisfy the Whitney's (a) condition: For any $a\in A_k$,
\[\lim_{A_l \ni x\to a}\delta( T_aA_k,  T_xA_l)=0.\]
\end{lemma}

\subsection{Domain Partitioning}
Now we are ready to give the partitioning of $\Omega$. Define
\[\mathrm{sign}(t):=\begin{cases}
1,&t>0,\\
0,&t=0,\\
-1,&t<0.
\end{cases}
\]
Then we define $\sigma_i(x):=\mathrm{sign}(z_i(x))$ for $x\in\Omega,i=1,\cdots,m-1$, and $\sigma(x):=(\sigma_1(x),\cdots,\sigma_{m-1}(x))$. 

For $s=(s_1,\cdots,s_{m-1})\in\{-1,0,1\}^{m-1}$, we define a function $f^s(x)$ by the following evaluation process.
\begin{align*}
    z_1^s&:=f_1(x),\\
    z_2^s&:=f_2(x,s_1z_1),\\
    &\cdots\\
    z_i^s&:=f_i(x,s_1z_1,s_2z_2,\cdots,s_{i-1}z_{i-1}),\\
    &\cdots\\
    z_m^s&:=f_m(x,s_1z_1,s_2z_2,\cdots,s_{m-1}z_{m-1}),\\
    f^s(x)&:=z_m^s.
\end{align*}
Then we define 
\[Q_{s}:=\{x\in\Omega:\text{For every } i=1,\cdots,m-1,\sigma_i(x)=s_i\}.\]
The closure of $Q_{s}$ is 
\[\bar{Q}_{s}=\{x\in\Omega:\text{For every } i=1,\cdots,m-1,\sigma_i(x)=s_i\text{ or }\sigma_i(x)=0\}.\]
For $x\in\bar{Q}_{s},f(x)=f^s(x)$, and $f^s(x)$ is analytic on $\bar{Q}_{s}$. Moreover, $Q_s$'s and $\bar{Q}_s$'s are compact semi-analytic sets. Now we have partitioned the domain $\Omega$ into the disjoint union of $Q_s$'s.

\subsection{Limit Behaviour of \texorpdfstring{$\delta$}{delta}-Stationary Points as \texorpdfstring{$\delta\to 0$}{delta->0}}
Next, we present the limit behaviour of a sequence of $\delta$-stationary points as $\delta\to 0$ in view of their locations. We call $x_0\in\Omega$ a \emph{limit stationary point} if there exists a point sequence $\{x_k\}_{k\geq 1}$ in $\Omega$ and a corresponding nonnegative sequence $\{d_k\}_{k\geq 1}$, such that 
\begin{itemize}
    \item $\lim_{k\to\infty}x_k=x_0$, 
    \item $\limsup_{k\to\infty}d_k=0$, and
    \item for $k=1,2,\cdots$, $x_k$ is an $d_k$-stationary point.
\end{itemize}

Suppose $x_0\in\Omega$ is such a limit stationary point. Since $\{x_k\}$ is an infinite sequence, there is at least one $s\in\{-1,0,1\}^{m-1}$ such that for infinitely many $k$, $\sigma(x_k)=s$. Without loss of generality we assume that $\sigma(x_k)=s$ for all $k\geq 1$. This implies that $x_k\in Q_{s}$ and $x_0\in\bar{Q}_s$.

We first state the relation between the location of the limit stationary point and the location of the $\delta$-stationary point sequence. By \Cref{prop:semi-analytic-Whitney}, $\bar{Q}_{s}$ can be decomposed into a disjoint union of a finite collection $\mathcal{A}$ of connected smooth semi-analytic subsets. There is subset $A\in\mathcal{A}$ so that $x_0\in A$ and $X\in\mathcal A$ such that $X$ contains infinitely many items of $\{x_k\}$. Again without loss of generality we assume that $x_k\in X$ for all $k$. Since $\lim_{k\to\infty}x_k=x_0$, by \Cref{prop:semi-analytic-Whitney} we know that either $A=X$, or $A\subseteq \partial X$. In either case the Whitney's (a) condition holds for $A$ and $X$, so $\lim_{k\to\infty}\delta( T_{x_0}A; T_{x_k}X)=0$.

Let the graph of $f^s$ on $A$ be
\[A':=\{(x,y)\in\R^{n+1}:x\in A,y-f^{s}(x)=0\}.\]
Since $A$ is a connected smooth semi-analytic set in $\R^n$ and $f^{s}(x)$ is analytic on $\bar{Q}_s\supseteq A$, $A'$ is a connected smooth semi-analytic set in $\R^{n+1}$.

Our next step is to construct a set $A^*\subseteq A$ consisting of points in $A$ whose gradient on manifold $A$ is $0$. Let $\pi:\R^n\times\R\to\R$ be the projection $\pi(x,y)=y$. 
\[A^*:=\{x\in A:\rank_{(x,f^s(x))}(\pi\vert_{A'})=0\}.\]
By \Cref{prop:semi-analytic-set-rank}, $A^*$ is semi-analytic since $A'$ is an analytic manifold.

A key observation is that
\begin{lemma}\label{lemma:x0-in-A*}
$x_0\in A^*$.
\end{lemma}

The lemma says that limit stationary point $x_0$ has zero gradient on manifold $A$.
\begin{proof}
We prove the lemma by contradiction. For mapping $g$, let $J_g(x)$ be the Jacobi matrix of $g$ at $x$. Suppose that $x_0\notin A^*$, which means $\rank_{(x_0,f^s(x_0))}(\pi\vert_{A'})>0$. Furthermore, $\pi\vert_{A'}$ maps $A'$ to $\R$, a one-dimensional space, so $\rank_{(x_0,f^s(x_0))}(\pi\vert_{A'})\leq 1$. Thus
\[\rank_{(x_0,f^s(x_0))}(\pi\vert_{A'})=1.\]

Let $K=\dim A'$. Define mapping $\psi:A\to A'$ by $\psi(x)=(x,f^s(x))$ for $x\in A$. Note that $\pi\circ\psi=f^s$. Take a local chart $(U,\xi)$ of $A'$ at $(x_0,f^s(x_0))$, where $U$ is the homeomorphic neighborhood of $(x_0,f^s(x_0))$ and $\xi$ is the analytic coordinate transformation from $U$ to $\R^K$. Define $\hat{\xi}:=\xi\circ\psi$. Let $z:=\xi(x_0,f^s(x_0))=\hat{\xi}(x_0)\in\R^K$. Since $\xi$ is the analytic coordinate transformation, we have 
\[\rank_z(\pi\circ\xi^{-1})=\rank_{(x_0,f^s(x_0))}(\pi\vert_{A'})=1.\]
By the chain rule of differentials, we have
\[J_{\pi\circ\xi^{-1}}(z)=J_{\pi\circ\psi\circ\hat{\xi}^{-1}}(z)=\grad {f^s}(x_0)J_{\hat{\xi}^{-1}}(z)\neq 0_K^T.\] 
Therefore there is some $v\in\R^K$ such that $\grad {f^s}(x_0)J_{\hat{\xi}^{-1}}(z)v<0$. Let
\[u:=\frac{J_{\hat{\xi}^{-1}}(z)v}{\norm{J_{\hat{\xi}^{-1}}(z)v}}.\]
We have $u\in  T_{x_0}A$, $\norm{u}=1$, and $\grad {f^s}(x_0)u<0$. Note that $\norm{\grad {f^s}(x)}$ is bounded on compact set $\bar{Q}_s$, so we have 
\begin{align*}
    &-\limsup_{k\to\infty}d_k\\
\leq&\liminf_{k\to\infty}\inf_{u'\in T_{x_k}\Omega}\frac{\partial f^s(x_k)}{\partial u'}\\
\leq&\liminf_{k\to\infty}\min_{u'\in  T_{x_k}X}\frac{\partial f^s(x_k)}{\partial u'}\\
=&\liminf_{k\to\infty}\min_{u'\in  T_{x_k}X,\norm{u'}=1}\grad {f^s}(x_k)u'\\
\leq&\liminf_{k\to\infty}\min_{u'\in  T_{x_k}X,\norm{u'}=1}\left(\grad {f^s}(x_k)u+\norm{\grad {f^s}(x_k)}\cdot\norm{u'-u}\right)\\
\leq &\liminf_{k\to\infty}\left(\grad {f^s}(x_k)u+\norm{\grad {f^s}(x_k)}\cdot \max_{u''\in  T_{x_0}A,\norm{u''}=1}\min_{u'\in  T_{x_k}X,\norm{u'}=1}\norm{u'-u''}\right)\\
=&\liminf_{k\to\infty}\left(\grad {f^s}(x_k)u+\norm{\grad {f^s}(x_k)}\cdot\delta\left( T_{x_0}A,  T_{x_k}X\right)\right)\\
=&\liminf_{k\to\infty}\grad {f^s}(x_k)u\\
=& \grad {f^s}(x_0)u<0.
\end{align*}

This contradicts the fact $\limsup_{k\to\infty}d_k=0$. Therefore $x_0\in A^*$.
\end{proof}

By \Cref{prop:semi-analytic-set-properties}, there are finitely many connected components of $A^*$. Assume that $\hat{A}\subseteq A^*$ is the connected component containing $x_0$. We have the following property.
\begin{lemma}\label{lemma:constant-on-hatA}
$f(x)$ is a constant on $\bar{\hat{A}}$, the closure of $\hat{A}$.
\end{lemma}
\begin{proof}
Let $x''$ be any point in $\hat{A}$. We need to show that for every $x\in\bar{\hat{A}}$, $f(x'')=f(x)$. Note that $\bar{\hat{A}}$ is a semi-analytic set. By \Cref{prop:semi-analytic-set-properties}, there is a semi-analytic curve $\varphi(t):[0,1]\to\bar{\hat{A}}$ such that $\varphi(0)=x,\varphi(1)=x''$, and $\varphi(t)\in\bar{\hat{A}}$ for $t\in(0,1]$. $\varphi(t)$ is analytic everywhere except finitely many points, so without loss of generality we assume that $\varphi(t)$ is analytic on $(0,1)$. It suffices to prove that $f^s(\varphi(0))=f^s(\varphi(1))$. 

For any mapping $g$, let $J_g(x)$ be the Jacobi matrix of $g$ at $x$. We have
\[f^s(\varphi(0))-f^s(\varphi(1))=\int_0^1\grad {f^s}(\varphi(t))J_{\varphi}(t)dt.\] 
Since $\varphi$ is a curve, the Jacobi matrix is a tangent vector, i.e., $J_{\varphi}(t)\in T_{\varphi(t)}A$. We claim that for every $x\in A^*$ and every $u\in T_x A$, $\grad {f^s}(x)u=0$. Then it follows that $f^s(\varphi(0))-f^s(\varphi(1))=0$, therefore $f^s(x)=f^s(x'')$, i.e., $f(x)=f(x'')$.

Now we prove the claim. Let $K=\dim A'$. Define mapping $\psi:A\to A'$ by $\psi(x)=(x,f^s(x))$ for $x\in A$. Note that $\pi\circ\psi=f^s$. Take a local chart $(U,\xi)$ of $A'$ at $(x,f^s(x))$, where $U$ is the homeomorphic neighborhood of $(x,f^s(x))$ and $\xi$ is the analytic coordinate transformation from $U$ to $\R^K$. Define $\hat{\xi}:=\xi\circ\psi$. Let $z:=\xi(x,f^s(x))=\hat{\xi}(x)\in\R^K$. Recall the definition of $A^*$,
\[\rank_z\left(\pi\vert_{A'}\right)=0.\]
Since $\xi$ is the analytic coordinate transformation, we have 
\[\rank_z(\pi\circ\xi^{-1})=\rank_z(\pi\vert_{A'})=0.\] 
By the chain rule of differentials, we have
\[J_{\pi\circ\xi^{-1}}(z)=J_{\pi\circ\psi\circ\hat{\xi}^{-1}}(z)=\grad {f^s}(x)J_{\hat{\xi}^{-1}}(z)= 0_K^T.\] 
For every $u\in T_x A$, it can be written in $J_{\hat\xi^{-1}}(z)v$ for some $v\in\R^K$, Thus
\[\grad f^s(x)u=\grad {f^s}(x)J_{\hat{\xi}^{-1}}(z)v=0.\]
\end{proof}

\subsection{Isolation Property of Local Minimum Point \texorpdfstring{$x^*$}{x*}}

We have shown that, for every $s\in\{-1,0,1\}^{m-1}$, there are finitely many $A$'s in the corresponding decomposition $\mathcal A$ of $\bar{Q}_s$. We have also constructed semi-analytic subset $A^*$ from each $A$. Let $P$ denote the union of all these $A^*$'s. $P$ is still a semi-analytic set, and has finitely many connected components by \Cref{prop:semi-analytic-set-properties}. Since $\Omega=\bigcup_{s\in\{-1,0,1\}^{m-1}}Q_s$, any limit stationary point $x_0\in\Omega$ is in $P$.

Now suppose stationary point $x^*\in\Omega$ is the unique minimum point in $B(x^*,r)\cap\Omega$ for some radius $r$. Observe that any stationary point is also a limit stationary point, so $x^*\in P$. We have the following lemma.
\begin{lemma}\label{lemma:x*-is-isolated}
$x^*$ is an isolated point of $P$. In other words, there is some $r_1>0$ such that $B(x^*,2r_1)\cap P=\{x^*\}$.
\end{lemma}
\begin{proof}
Suppose for the contrary that $x^*$ is not an isolated point of $P$. Note that there are only finitely many $A^*$'s and finitely many $\hat{A}$'s. In addition, $P$ is the union of $\hat{A}$'s. Thus there exists some infinite set $\hat{A}$ such that $x^*$ is a cluster point of $\hat{A}$. However, by \Cref{lemma:constant-on-hatA} we know that for any $x\in \hat{A}$, $f(x)=f(x^*)$. Thus there exists an infinite sequence $x_k\to x^*$ in $\hat{A}$ such that $f(x_k)=f(x^*)$ for any $k$. That contradicts the fact that $x^*$ is the unique local minimum point in $B(x^*,r)\cap\Omega$. Therefore $x^*$ must be an isolated point of $P$. 
\end{proof}

At last, we prove \Cref{thm:analytic-abs-delta-SP-converge} by contradiction. Take $r_1$ in \Cref{lemma:x*-is-isolated}. Suppose that there is some $\varepsilon_0>0$ such that for every $\delta>0$, there exists $x'\in (B(x^*,r_1)\setminus B(x^*,\varepsilon_0))\cap\Omega$, such that $x'$ is a $\delta$-stationary point. Equivalently, there is an infinite point sequence $\{x_k\}$ in $(B(x^*,r_1)\setminus B(x^*,\varepsilon_0))\cap\Omega$, and a corresponding sequence $\{d_k\}$, such that $\limsup_{k\to\infty}d_k=0$, and for $k=1,2,\cdots$, $x_k$ is an $d_k$-stationary point. Take any convergent subsequence of $\{x_k\}$, and assume it converges to some point $x_0$. By \Cref{lemma:x0-in-A*}, $x_0\in P$ and $\norm{x_0-x^*}\in[\varepsilon_0,r_1]$, which contradicts the fact $B(x^*,2r_1)\cap P=\{x^*\}$.

Now we finally complete the proof of \Cref{thm:analytic-abs-delta-SP-converge}.

\begin{remark}
In the proof of \Cref{thm:analytic-abs-delta-SP-converge}, only the last part uses the fact that $x^*$ is an isolated local minimum point. In fact, if we give up a neat statement, \Cref{thm:analytic-abs-delta-SP-converge} can be generalized. Let $L$ be a connected set of local minimum points that is isolated from other stationary points. Then the set of $\delta$-stationary points in some neighborhood of $L$ is shrinking to $L$ as $\delta\to 0$.
\end{remark}
\section{Proof of Geometric Properties of Stationary Points}\label{sec:locus-of-SP}
In this section, we use \Cref{lemma:x0-in-A*} and \Cref{lemma:constant-on-hatA} to prove \Cref{prop:geometric-char-of-sp}. We have shown that all the limit-stationary points (also, stationary points) are contained in a semi-analytic set $P$, which has some good properties:
\begin{enumerate}
    \item $P$ is a semi-analytic set, and can be decomposed into the union of a finite collection of connected smooth semi-analytic subsets.
    \item $P$ has finitely many connected components, and $f$ is constant on each component.
\end{enumerate}

However, one may notice that, not every point in $P$ is a limit-stationary point. Inspired by \cite{DBLP:journals/siamjo/BolteDL07} and \cite{coste2000introduction}, we introduce sub-analytic sets and sub-analytic functions, and give a better characterization on the locus of all the limit-stationary points, and stationary points, respectively.

We say a set $A\subseteq\R^n$ to be \emph{sub-analytic} if each point $x\in \R^n$ admits a neighborhood $U$, such that 
\[A\cap U=\pi_{n+m,n}(B),\]
where $B$ is a bounded semi-analytic subset of $\R^{n+m}$ for some $m$, and 
\[\pi_{n+m,n}(x_1,\cdots,x_{n+m})=(x_1,\cdots,x_n)\] 
is the canonical projection from $\R^{n+m}$ to $\R^n$.

\begin{proposition}[\cite{Semianalytic-and-subanalytic-sets}, Section 3]\label{prop:sub-analytic-set-properties}
We have the following properties of sub-analytic sets:
\begin{itemize}
    \item Every semi-analytic set is sub-analytic.
    \item Sub-analytic sets are closed under finite union and intersection.
    \item The closure, interior, and complement of a sub-analytic set is sub-analytic.
    \item The Cartesian product of two sub-analytic sets is sub-analytic.
    \item If $A\subseteq \R^n$ is a sub-analytic set, then $\pi_i(A)$ is sub-analytic, here \[\pi_i(x_1,\cdots,x_n)=(x_1,\cdots,x_{i-1},x_{i+1},\cdots,x_n).\]
\end{itemize}
\end{proposition}

For a mapping $f:\Omega\to \R^m$ defined on $\Omega\subseteq \R^n$, the \emph{graph} of $f$, denoted by $\mathrm{graph} f$, is defined by
$$\mathrm{graph} f:=\{(x,y):x\in \Omega,y=f(x)\}\subseteq\R^{n+m}.$$
We say $f$ is a sub-analytic mapping, if $\mathrm{graph} f$ is a sub-analytic subset of $\R^{n+m}$. Similarly we say a function $f:\Omega\to \R$ is sub-analytic if $\mathrm{graph} f$ is a sub-analytic subset of $\R^{n+1}$. We will use the following facts:
\begin{proposition}[\cite{denkowska1979certaines}]\label{prop:sub-analytic-mapping-properties}
\begin{itemize}
    \item If $f:A\to B$ is a sub-analytic mapping, and $E$ is a bounded sub-analytic subset of $A$, then $f(E)$ is a sub-analytic set.
    \item If $f:A\to B,g:B\to C$ are both sub-analytic mappings, and $f$ maps bounded sets to bounded sets, then the composition $g\circ f$ is sub-analytic.
    \item Let $A$ be a sub-analytic set. A function $f:A\to\R$ is sub-analytic if and only if the epigraph of $f$, defined as $\mathrm{epi}f:=\{(x,y)\in A\times\R:y\geq f(x)\}$, is sub-analytic.
\end{itemize}
\end{proposition}

An important observation in \cite{coste2000introduction} is that ``all first-order formula is definable'', in other words, sub-analytic sets are able to express quantifiers.

\begin{lemma}\label{lemma:sub-analytic-first-order-formula}
If $A\subseteq \R^n,B\subseteq \R^m,C\subseteq \R^{m+n}$ are sub-analytic sets, then $$\begin{aligned}
\{x\in A:\exists y\in B,(x,y)\in C\}&\text{ and}\\
\{x\in A:\forall y\in B,(x,y)\in C\}&
\end{aligned}$$ are both sub-analytic sets.
\end{lemma}
\begin{proof}
Observe that $\{x\in A:\exists y\in B,(x,y)\in C\}=\pi_{n+m,n}((A\times B)\cap C)$. As sub-analytic sets are closed under complement operation, and therefore under set minus operation, we can similarly write $\{x\in A:\forall y\in B,(x,y)\in C\}$ as $A\setminus\pi_{n+m,n}((A\times B)\setminus C)$. By \Cref{prop:sub-analytic-set-properties}, we immediately get that $\{x\in A:\exists y\in B,(x,y)\in C\}$ and $\{x\in A:\forall y\in B,(x,y)\in C\}$ are both sub-analytic.
\end{proof}

Now we consider $f$ defined on a compact semi-analytic set $\Omega\subseteq \R^n$, as previously described. We show that:
\begin{lemma}\label{lemma:f-is-sub-analytic}
$f(x)$ and $G_f(x):=\inf_{s}\frac{\partial f(x)}{\partial s}=\inf_{x'\neq x}\frac{Df(x,x')}{\norm{x'-x}}$ are both sub-analytic functions.
\end{lemma}
\begin{proof}
Recall that $f(x)$ is defined with $f_1,\cdots,f_m$, by the evaluation process that for $i=1,\cdots,m$, $z_i(x):=f_i(x,|z_1(x)|,|z_2(x)|,\cdots,|z_{i-1}(x)|)$, and $f(x):=z_m(x)$.

$z_1(x)=f_1(x)$ is analytic on $\Omega$, so $\mathrm{graph} z_1$ is an semi-analytic subset of $\R^{n+1}$, thus $z_1(x)$ is sub-analytic. Clearly the absolute-value function is also sub-analytic. Additionally, because of continuity, each $z_i(x)$ is bounded on compact set $\Omega$. Therefore if $z_1(x),\cdots,z_{k}(x)$ are sub-analytic, by \Cref{prop:sub-analytic-mapping-properties} it follows that $z_{k+1}(x)=f_{k+1}(x,|z_1(x)|,|z_2(x)|,\cdots,|z_{k}(x)|)$ is also sub-analytic. By induction we get that $f(x)=z_m(x)$ is a sub-analytic function.

Then with the help of \Cref{lemma:sub-analytic-first-order-formula}, we prove that $G_f(x)$ is a also sub-analytic function. For convenience, we first show that the Dini derivative $Df(\cdot,\cdot)$ is sub-analytic. To keep the notations simple, we assume that every direction of $x$ in $\Omega$ can be expressed by a segment starting at $x$ and contained in $\Omega$. We write $\mathrm{graph}{Df(\cdot,\cdot)}$ as a first-order logic expression
\begin{gather*}
    \{(x,x',\delta)\in\Omega\times\Omega\times\R:\\\forall \epsilon>0,\exists r>0,\forall t>0, t>r\vee(\delta-\epsilon)t< f(x+t(x'-x))-f(x)<(\delta+\epsilon)t\},
\end{gather*}
which is sub-analytic. Hence $Df(\cdot,\cdot)$ is sub-analytic.

The epigraph of $G_f(x)$, 
\[\mathrm{epi} G_f:=\left\{(x,\delta)\in\Omega\times\R:\delta\geq  G_f(x)=\inf_{x'\neq x}\frac{Df(x,x')}{\norm{x'-x}}\right\}\]
can be written as 
\[\{(x,\delta)\in\Omega\times\R:\exists x'\in \Omega, \norm{x'-x}>0\wedge Df(x,x')\leq \delta\norm{x'-x}\}.\]
Therefore $G_f(x)$ is a sub-analytic function, by \Cref{prop:sub-analytic-mapping-properties}.
\end{proof}

We are ready to prove our characterization result on stationary points in \Cref{prop:geometric-char-of-sp}. Let $S:=\{x\in\Omega:x\text{ is a stationary point of }f\}$, we have
\[S=\{x\in\Omega: G_f(x)\geq 0\}=\{x\in\Omega:\forall y<0, G_f(x)>y\}.\]
By \Cref{lemma:f-is-sub-analytic} and \Cref{lemma:sub-analytic-first-order-formula}, $S$ is a sub-analytic set.

By the theorem on sub-analytic stratification \cite{denkowska1979certaines,lojasiewicz1995semi} similar to \Cref{thm:semi-analytic-stratification}, we know that $S$ can be decomposed as the union of a finite collection $\{S_i\}$ of sub-analytic subsets, so that each $S_i$ is a connected analytic manifold.

Note that $S\subseteq P$, where $P$ is the semi-analytic set described at the beginning of this section. So each $S_i$ must be a subset of some connected component $P'$ of $P$, and therefore $\bar{S_i}\subseteq \bar{P'}$. Recall that $f$ is constant on the closure of every connected component of $P$, therefore for any $x_1,x_2\in \bar{S_i}$, we have $f(x_1)=f(x_2)$.

Suppose $x^*\in\Omega$ is a non-isolated local minimum point of $f$, then there exists a connected neighborhood $U$ of $x^*$, so that $f(x)\geq f(x^*)$ for all $x\in U$. $x^*$ is clearly a stationary point, so $x^*\in S$. $S$ has finitely many connected components. Furthermore, for each connected component $S'$, $f$ is constant on $\bar{S'}$. If $x^*\in \bar{S'}$, we have $f(x)=f(x^*)$ for $x\in S'$; otherwise, we have $\inf_{x\in S'}\lVert x-x^*\rVert>0$. Therefore there is $r>0$ such that for all $x\in S\cap B(x^*,r)$, $f(x)=f(x^*)$. Let $U^*=U\cap B(x^*,r)$, then for every $x\in S\cap U^*$, there is a neighborhood $U'$ of $x$, such that $U'\subseteq U$. Then we have $f(x)=f(x^*)\leq \inf_{x'\in U'}f(x')$, so $x$ is also a local minimum point. Since $x^*$ is a non-isolated local minimum point, $S\cap U^*$ is a connected infinite set containing $x^*$. Thus every $x\in S\cap U^*$ is a non-isolated local minimum point. Now we complete the proof.

\begin{remark}
The set of all limit-stationary points in $\Omega$, denoted by $S^{\lim}$, can be written as
\[S^{\lim}=\{x\in\Omega:\forall\delta>0,\forall r>0,\exists x'\in\Omega,\lVert x-x'\rVert<r\wedge G_f(x')>-\delta\}.\]
Thus all limit-stationary points also form a sub-analytic set, and have similar geometric properties to those stated in \Cref{prop:geometric-char-of-sp}.
\end{remark} 

\section{Convergence Stability}\label{sec:convergence-stability}
In this section, we use \Cref{thm:analytic-abs-delta-SP-converge} to study a natural stability concept called convergence stability and prove \Cref{thm:analytic-max-is-convergence-stable}. 

We consider a specific optimization algorithm $\A$ for function $f$. $\A$ chooses some initial point $x_0$ and makes many steps of an optimization procedure. Each step yields a point $x_k$. When certain conditions are satisfied, $\A$ terminates; otherwise, $\A$ never stops. Although a perfect algorithm should find an exact stationary point, it is impossible due to limited precision and time. Thus, most algorithms aim to find an approximate stationary point \cite{bagirov2014introduction}. Precisely, given a precision $\delta$ as the input parameter, a reasonable algorithm should
\begin{itemize}
    \item either terminate at a $\delta$-stationary point,
    \item or generate an infinite point sequence $\{x_k\}$ with $\delta$-stationary points as its cluster\footnote{A cluster point $x'$ of set $\{x_k\}$ is defined as for any $r>0$, for infinitely many $k$, $x_k\in B(x',r)\setminus\{x'\}$.} points.
\end{itemize}

Our goal is to define a stability concept about approximate stationary points. Naturally, stability means a $\delta$-stationary point should be close to a stationary point when $\delta$ is close to $0$; otherwise, a slight difference on $\delta$ may lead to a significant change of the found solution. Such changes are ubiquitous, especially in nonconvex optimization problems. Moreover, the changes may result in a worse $f$ value or other worse properties (e.g., in deep learning, test accuracies vary for models even with the same training loss). 

To focus on $\delta$-stationary points near a certain stationary point, we only inspect a neighborhood of the stationary point $x^*$. When the optimization algorithm $\A$ starts from a point $x_0$ at the stationary point $x^*$, the following two statements both describe the stability of a stationary point $x^*$:
\begin{itemize}
    \item For any sufficiently small precision parameter $\delta$, $\A$ will find a $\delta$-stationary point \emph{near} $x^*$  regardless of the choices of the initial point in the neighborhood of $x^*$.
    \item Any $\delta$-stationary point in the neighborhood gets close to $x^*$ when $\delta\to 0$.
\end{itemize}

This observation motivates the formal definition of convergence stability. Specifically, for algorithm $\A$, we denote the starting point by $x_0$, and the point after the $k$-th iteration by $x_k$. For convenience, we assume that if $\A$ finds a $\delta$-stationary point $x_{k_0}$, then for all $k>k_0$, it sets $x_k=x_{k_0}$, so that $\{x_k\}$ is always an infinite sequence. A \emph{limit point} $x'$ of $\{x_k\}$ is then defined as: For any $r>0$, for infinitely many $k$, $x_k\in B(x',r)$.

We need to restrict types of optimization algorithms to obtain nontrivial results. In many practical scenarios \cite{bagirov2014introduction}, we observe the following three properties:

\begin{enumerate}
    \item The objective function $f(x)$ decreases after each iteration, i.e., for all $k=0,1,2,\cdots$, 
    \[f(x_{k+1})\leq f(x_{k}).\]
    \item Every limit point of $\{x_k\}$ is a $\delta$-stationary point, where $\delta$ can be controlled by $\mathcal A$'s input parameters.
    \item The step size $\norm{x_k-x_{k+1}}$ is limited, $i.e.$, there is $\lambda>0$ such that for all $k=0,1,2,\cdots$, $\norm{x_k-x_{k+1}}\leq\lambda$. If we assume $f(x)$ is $L$-Lipschitz, then this implies that for any $\theta\in[0,1]$,
    \[f(\theta x_{k+1}+(1-\theta)x_{k})\leq f(x_{k})+L\lambda.\]
\end{enumerate}

Our analysis needs $\{x_k\}$ to have the following more relaxed properties.

\begin{definition}\label{def:commonalgorithmproperties}
We say $\{x_k\}_{k\geq 0}$ is 
\begin{itemize}
    \item \emph{ultimately decreasing}, if every limit point $\bar{x}$ of $\{x_k\}$ satisfies $f(\bar{x})\leq f(x_0)$.
    \item $\delta$-\emph{result-stationary}, if every limit point of $\{x_k\}$ is $\delta$-stationary point.
    \item $\lambda$-\emph{path-bounded}, if for any $ k=0,1,\cdots$ and any $\theta\in[0,1]$, $f(\theta x_{k+1}+(1-\theta)x_{k})\leq f(x_0)+\lambda$.
\end{itemize}
\end{definition}

We say an algorithm $\mathcal A$ has these properties, if the point sequence $\{x_k\}$ generated by $\mathcal A$ is guaranteed to always (or almost surely\footnote{Throughout this section, the term ``almost surely" (with respect to a probability measure) only accounts for the randomness inside the algorithm. Choices of initial points are not taken into consideration.}) satisfy these properties, regardless of the starting point $x_0$. The ultimate decrease property is very reasonable as long as an optimization algorithm does try to find a local minimum point. The $\delta$-result-stationary property is usually theoretically proved for many effective optimization algorithms. The $\lambda$-path-bounded property is usually satisfied in practice\footnote{For instance, we consider the descent optimization process on a continuous function $f$. By setting the step size sufficiently small, we will have path-bounded property due to uniform continuity by Cantor's theorem.}, and can be easily guaranteed for the descent optimization process on an $L$-Lipschitz function $f$ by limiting the step size $\lVert x_k-x_{k+1}\rVert< \lambda/L$. Furthermore, the $\lambda$-path-bounded property allows the analysis to tolerate a small magnitude of numerical errors. Hence such requirements on optimization algorithms are more practical.

Now we present the definition of convergence stability. We use notation $A_{\delta, \lambda}$ to denote an optimization algorithm satisfying ultimate decreasing, $\delta$-result-stationary and $\lambda$-path-bounded properties.

\begin{definition}\label{def:converge-stab}
A stationary point $x^*\in \Omega$ is called $(r_0,\lambda)$-\emph{convergence-stable} if for every $\varepsilon>0$, there exists $\delta>0$,
such that for any $\delta_1\in(0,\delta),\lambda_1\in(0,\lambda)$ and any initial point $x_0\in B(x^*,r_0)$, the point sequence $\{x_k\}_{k\geq 0}$ generated by an \emph{arbitrary} algorithm $\mathcal A_{\delta_1,\lambda_1}$ almost surely satisfies
\[\limsup_{k\to\infty}\norm{x_k-x^*}< \varepsilon.\]
If we simply say $x^*\in \Omega$ is \emph{convergence-stable}, it means that such $r_0$ and $\lambda$ exist but do not necessarily have an explicit form.
\end{definition}

The definition of convergence stability may be not intuitive, so we give a few simple examples to demonstrate it.

\begin{example}
Let $f(x)=x^2, x\in[-1,1]$. We want to show that stationary point $x^*=0$ is $(r_0,\lambda)$-convergence-stable for some $r_0$ and $\lambda$. $f'(x)=2x$. Thus $\delta$-stationary points forms an interval $[\Tilde{x}_-(\delta),\Tilde{x}_+(\delta)]$, where
\[\Tilde{x}_\pm(\delta)=\pm\frac{\delta}{2}.\]
By definition, any algorithm $\A_{\delta_1,\lambda_1}$ with $\delta_1\in(0,\delta)$, $\lambda_1>0$ and $x_0\in[-1,1]$ can find a solution in $[\Tilde{x}_-(\delta),\Tilde{x}_+(\delta)]$. As $\delta\to 0$, $\norm{\Tilde{x}_\pm(\delta)}\to 0$. Thus for any $r_0\in(0,1]$ and $\lambda>0$, $x^*=0$ is convergence-stable.
\end{example}

We also show a nonexample:
\begin{example}\label{ex:nonex-for-convergence-stable}
Let $f(x)=|x+1|+|x-1|, x\in[-2,2]$. We prove that $x^*=0$ is not convergence-stable. It suffices to note that any $\delta$-stationary point with $\delta<2$ lies in $(-1,1)$. Thus for any $r>0$ and $\lambda>0$, set $\epsilon=\min\{1,r/2\}$. An initial point $x_0=0.9r$ will make some algorithms get stuck and terminate. Thus $x^*$ fails the definition of convergence stability.
\end{example}

\Cref{ex:nonex-for-convergence-stable} implies that convergence stability is a property mainly determined by the stationary point $x^*$ itself. In other words, it makes the weakest assumptions on optimization algorithms and guarantees for \emph{any} ultimately decreasing, $\delta$-result-stationary and $\lambda$-path-bounded algorithm, the optimization solution cannot escape from a neighborhood of $x^*$ and the algorithm will always find a solution near $x^*$ once it goes into that neighborhood.

Now we turn to \Cref{thm:analytic-max-is-convergence-stable}. To demonstrate how this result comes about, we have the following reasoning. What kinds of stationary points could be convergence-stable? A saddle point cannot be stable: The descent process will never come back if we perturb it to a descent direction of the saddle point. So the stationary point must be a local minimum point. If a local minimum point is not isolated, i.e., encompassed by infinitely many stationary points, the optimization process may get stuck at another stationary point nearby. The two scenarios describe a necessary condition for convergence stability. Through \Cref{thm:analytic-abs-delta-SP-converge}, this is also proved to be the sufficient condition for function $f$. Now we state the formal proof of \Cref{thm:analytic-max-is-convergence-stable}.

\begin{proof}[Proof of \Cref{thm:analytic-max-is-convergence-stable}]
To avoid verbose notations, every notation $B(x,r)$ in the proof should be replaced by $B(x,r)\cap\Omega$. Such replacement does not affect the correctness of the proof.

\emph{Necessity}: Assume that $x^*$ is convergence-stable. Then there is some $r_0>0$ such that for any $x_0\in B(x^*,r_0)$ and $\varepsilon>0$, there exists sufficiently small $\delta,\lambda$ such that the sequence $\{x_k\}_{k\geq 0}$ generated by $\mathcal{A}_{\delta,\lambda}$ with initial point $x_0$ satisfies that every limit point $\bar{x}$ of $\{x_k\}$ is in $B(x^*,\varepsilon)$. As $\mathcal A_{\delta,\lambda}$ is ultimately decreasing, we have $f(\bar{x})\leq f(x_0)$. Since $\varepsilon$ can be arbitrarily small, by the continuity of $f$, we have $f(x^*)\leq f(x_0)$. Furthermore, $f(x^*)=\min_{x\in B(x^*,r_0)}f(x)$ since $x_0$ is arbitrarily chosen. So $x_0$ is a minimum point in $B(x^*,r_0)$.

To prove the necessity, we claim that for any $x'\neq x^*$ in $B(x^*,r_0)$, the strict inequality $f(x')>f(x^*)$ holds. Suppose otherwise that $f(x')\leq f(x^*)$, then combining it with $f(x^*)=\min_{x\in B(x^*,r_0)}f(x)$ we have $f(x')=\min_{x\in B(x^*,r_0)}f(x)$. It follows that $x'$ is a local minimum point. Then $\partial f(x')/\partial s\geq0$ holds for any valid direction $s$. Thus $x'$ is also a stationary point. But $x'$ cannot be a stationary point, because otherwise the descent procedure starting from $x'$ will stay at $x'$, which violates the assumption.

\emph{Sufficiency}: Assume that there is some $r_0>0$ such that 
\[f(x^*)=\min_{x\in B(x^*,r_0)}f(x),\]
and for $x'\neq x^*$ in $B(x^*,r_0)$, $f(x')>f(x^*)$.

By \Cref{thm:analytic-abs-delta-SP-converge}, without loss of generality, we may assume that $r_0$ is already sufficiently small, so that for any $\varepsilon>0$, there exists $\delta_0>0$, such that every $\delta_0$-stationary point in $B(x^*,r_0)$ is in $B(x^*,\varepsilon)$. For algorithm $\mathcal{A}_{\delta,\lambda}$ with parameter $\delta<\delta_0$, every limit point of the generated point sequence $\{x_k\}_{k\geq 0}$ is almost surely a $\delta_0$-stationary point. Thus we only need to prove that $x_k\in B(x^*,r_0)$ holds for all $k\geq 0$.

Take a subset of $B(x^*,r_0)$:
\[U_{[0.8r_0,0.9r_0]}:=\{x\in B(x^*,r_0):\lVert x-x^*\rVert\in[0.8r_0,0.9r_0]\}.\]

As $U_{[0.8r_0,0.9r_0]}$ is a compact set, we have $\bar{f}:=\min_{x\in U_{[0.8r_0,0.9r_0]}}f(x)>f(x^*)$. Let $\lambda_0:=\frac{\bar{f}-f(x^*)}2$, and let $\mathcal{A}_{\delta,\lambda}$ has parameter $\lambda<\lambda_0$.

By continuity of $f$, there exists $r_1\in(0,0.8r_0)$, such that for any $x\in B(x^*,r_1)$, $f(x)\leq f(x^*)+\lambda_0$. When the starting point $x_0$ is in $B(x^*,r_1)$, we claim that $x_k\in B(x^*,0.8r_0)$ for every $k>0$. That is because otherwise there exists $k_0>0$ and $\theta\in[0,1)$ such that $(1-\theta)x_{k_0}+\theta x_{k_0+1}\in U_{[0.8r_0,0.9r_0]}$. Therefore $f((1-\theta)x_{k_0}+\theta x_{k_0+1})\geq f(x^*)+2\lambda_0\geq f(x_0)+\lambda_0>f(x_0)+\lambda$, contradicting the fact that $\mathcal{A}_{\delta,\lambda}$ is $\lambda$-path-bounded.

In summary, we prove that for any $\varepsilon>0$, there exists $\delta_0>0$ such that for every $\delta\in(0,\delta_0)$, $\lambda\in(0,\lambda_0)$ and $x_0\in B(x^*,r_1)$, for $\mathcal{A}_{\delta,\lambda}$ starting from $x_0$, every limit point of the generated point sequence $\{x_k\}_{k\geq 0}$ is almost surely a $\delta_0$-stationary point in $B(x^*,r_0)$, so we have $x_k\in B(x^*,\varepsilon)$ for infinitely many $k$. Hence $x^*$ is $(r_1,\lambda_0)$-convergence-stable by definition.
\end{proof}

In the sufficiency proof of \Cref{thm:analytic-max-is-convergence-stable}, we actually present a quantitative approach to proving a stationary point to be $(r_1,\lambda_0)$-convergence-stable. We summarize it as the following corollary.
\begin{corollary}\label{prop:quan-result-of-convergence-stable}
Suppose $x^*$ is the unique minimum point of $f$ in $B(x^*,r)\cap\Omega$. Let
\[\lambda_0:=\frac{1}{2}\left(\min_{\norm{x-x^*}\in[0.8r,0.9r]} f(x)-f(x^*)\right),\]
and 
\[r_1:=\max\{r'\in(0,0.8r]:\text{For every } x\in B(x^*,r')\cap\Omega, f(x)\leq f(x^*)+\lambda_0\}.\]
Then $x^*$ is $(r_1, \lambda_0)$-convergence-stable.
\end{corollary}

Finally, we remark that since \Cref{thm:analytic-abs-delta-SP-converge} fails in \Cref{ex:nonex-for-converge-stable-C1} and \Cref{ex:nonex-for-converge-stable}, \Cref{thm:analytic-max-is-convergence-stable} fails as well in these examples. Indeed, an optimization algorithm may get stuck at a $\delta$-stationary point ($x_0$ in these examples) far away from the stationary point ($x^*$ in these example) and then terminate. Thus analyticity is also necessary for \Cref{thm:analytic-max-is-convergence-stable}.
\section{Discussion and Open Problems}\label{sec:discussion}
In this work, we prove the convergence lemma for piecewise analytic function $f$. Extending our technique, we characterize the geometric properties of the set of stationary points of $f$. Finally, we introduce convergence stability, and show an intuitive equivalent condition of convergence stability for $f$ via the convergence lemma. Our results reveal some understandings on nonconvex nonsmooth optimization, as well as the new optimization methodology of deep neural networks. Questions
raise themselves for future explorations. The next immediate open issues include the following
future topics.
\begin{enumerate}[fullwidth,listparindent=\parindent]
\item Further \emph{geometric} and \emph{qualitative} understandings on deep neural networks and nonsmooth nonconvex functions.
    
The main methods used to study deep neural networks and nonsmooth nonconvex functions are probability theory and analysis theory. Techniques such as NTK theory are developed to do convergence and performance analyses on various neural network models and optimizers \cite{DBLP:conf/icml/Allen-ZhuLS19,DBLP:conf/icml/DuLL0Z19,DBLP:conf/iclr/DuZPS19,DBLP:conf/nips/JacotHG18}. Theoretical analyses over gradient sampling methods are also developed \cite{Burke2002GradientSampling,BurkeRobustGradientSampling2005,DBLP:journals/siamjo/HosseiniU17,DBLP:journals/siamjo/Kiwiel07}. Many results are proved by heavy calculations and delicate inequalities. So these results are very quantitative and able to resolve problems for a specific kind of neural networks and a specific kind of optimization methods. Cooper's work \cite{DBLP:journals/simods/Cooper21} and our work provide another extremity: a qualitative understanding on very general neural networks via geometry theory. There are far fewer calculations. Even though not very quantitative, the results can apply to general neural networks and nonconvex functions. Combining understandings from both extremities, we will get a better understanding of the ``black magic'' behind deep neural networks and the difficulties arising from nonconvex nonsmooth optimization.
   
\item More \emph{quantitative} analysis on the convergence lemma and convergence stability.
   
Our work applies to very general object functions and optimization algorithms; the price for such generality is a lack of quantitative results. For strongly convex functions, the convergence lemma is quite accurate: The gradient norm is linear to the loss difference \cite{boyd2004convex}. Our result, however, only present the continuity property. As for convergence stability, if we wish to guarantee a $(\lambda,r)$-convergence-stable stationary point, how should we choose the step size and how large can $r$ be? \Cref{prop:quan-result-of-convergence-stable} presents a calculation method, but for most functions, such calculation is cumbersome. It would be interesting to find a specific class of important objective functions so that it is easier to do practical calculations on them.
   
\item \emph{Probabilistic considerations} on convergence stability.
   
Most previous theories, such as deep learning theory \cite{DBLP:conf/icml/Allen-ZhuLS19,DBLP:conf/icml/DuLL0Z19,DBLP:conf/iclr/DuZPS19,DBLP:conf/nips/JacotHG18} and smoothed analysis \cite{SpielmanCommunACM}, involve the probability into their settings. For instance, a minor error means a mean-zero Gaussian with a small variance in these theories, rather than a small positive radius in our work. Involving probability into settings would make further comparisons between these theories possible and thus help us more completely understand stability. Furthermore, probability is conducive to improving local results to global results.
\end{enumerate}

We hope our work makes a steppingstone toward a new understanding of the convergence theory and stability theory of optimization problems in general, especially of deep neural networks. 


\bibliographystyle{plain}
\bibliography{references}

\appendix
\section{Verifying (Counter-)\texorpdfstring{\Cref{ex:nonex-for-converge-stable}}{Example 3.2} for Convergence Stability}\label{app:nonexample-convergence-stable}
We do the verification step by step.
\begin{enumerate}[fullwidth, listparindent=\parindent]
    \item \emph{If $x>0$,
    \begin{align*}
    g'(x)&=x^{-2}e^{-1/x}\sin\left(1/x\right)-x^{-2}e^{-1/x},\\
    h'(x)&=x^{-2}e^{-1/x}\sin\left(1/x\right)-x^{-2}e^{-1/x}\cos\left(1/x\right).
\end{align*}}

    It follows by direct calculations.
    
    \item \emph{For both $g(x)$ and $h(x)$, the differential of every order at $x=0$ is $0$.}
    
    Note that for $x\neq 0$, the $n$th derivative of $g(x)$, denoted by $g^{(n)}(x)$, can be expressed in $\poly(1/x,\sin(1/x),\cos(1/x))e^{-1/|x|}$. Clearly $|g^{(n)}(x)/x|\to 0$ as $x\to 0$. Thus $g^{(n)}(0)=0$ for all $n$. The same fact about $h(x)$ can be proved by a similar argument.
    
    \item \emph{$g(x)$ and $h(x)$ belong to $C^\infty(\R)$.}
    
    It is clear that $g(x)$ and $h(x)$ are $C^\infty$ when $x\neq 0$. The case when $x=0$ is verified previously.
    
    \item \emph{$f(x)$ is finite compositions of $C^\infty$ mappings and the maximum function.}
    
    Note that $\min\{a,b\}=-\max\{-a,-b\}$, $|x|=\max\{x,-x\}$. In addition, $g(x)$ and $h(x)$ are $C^\infty$. The result follows.
    
    \item \emph{$f'_{-}(x)>0$ for $\in(0,1)$, so $f(x)$ is strictly increasing on $(0,1)$. These hold symmetrically on $(-1,0)$. So $x^*=0$ is the only stationary point in $B(0,1)$, and is the unique local minimum point.}
    
    Suppose $x>0$. $\sin(1/x)<0$ if and only if $x\in\left(\frac{1}{2k\pi}, \frac{1}{(2k-1)\pi}\right)$, $k=1,2,\ldots$ Thus
    \[f_-'(x)=\begin{cases}
    h'(x)-g'(x),& x\in\left(\frac{1}{2k\pi}, \frac{1}{(2k-1)\pi}\right),k=1,2,\ldots,\\
    -g'(x),& x\in\left(\frac{1}{2(k+1)\pi}, \frac{1}{2k\pi}\right), k=1,2,\ldots,\\
    \max\{0,h'(x)\}-g'(x),& x=\frac{1}{k\pi}, k=1,2,\ldots \end{cases}\]
    
    We discuss case by case to show that $f_-'(x)>0$.
    
    If $x\in\left(\frac{1}{2k\pi}, \frac{1}{(2k-1)\pi}\right)$, $f_-'(x)=h'(x)-g'(x)=x^{-2}e^{-1/x}(1-\cos(1/x))>0$.
    
    If $x\in\left(\frac{1}{2(k+1)\pi}, \frac{1}{2k\pi}\right)$, $f_-'(x)=-g'(x)=x^{-2}e^{-x}(1-\sin(1/x))>0$.
    
    If $x=\frac{1}{(2k-1)\pi}$, $f_-'(x)=h'(x)-g'(x)=2x^{-2}e^{-1/x}>0$.
    
    If $x=\frac{1}{2k\pi}$, $f_-'(x)=-g'(x)=x^{-2}e^{-1/x}>0$.
    
    Note that $f(x)$ is a symmetric function, so $f_+'(x)<0$ for $x\in(-1,0)$.
    
    \item \emph{\Cref{thm:analytic-abs-delta-SP-converge} does not hold for $x^*=0$}
    
    Observe that for $k=1,2\ldots$ and $t\in(0,\pi)$, 
\begin{align*}
    f'\left(\frac{1}{2k\pi-t}\right)&=(2k\pi-t)^2e^{2k\pi-t}(1-\cos(2k\pi-t))\\
    &=(2k\pi-t)^2e^{2k\pi-t}(1-\cos(t))\\
    &\leq \frac{1}{2}(2k\pi)^2e^{2k\pi}t^2=C(k)t^2.
\end{align*}
Take any $r_0>0$. Let $k=\left\lceil\frac1{2\pi r_0}\right\rceil+1$ and $\epsilon=\frac1{2k\pi+1}$. For any $\delta>0$, let $t=\min\{\sqrt{C^{-1}(k)\delta},1\}$. Set $x_0=\frac1{2k\pi-t}\in B(0,r_0)$. Then $x_0$ permanently because $x_0$ is a $\delta$-stationary point. However, $x_0=\frac1{2k\pi-t}\notin B(0,\epsilon)$.
\end{enumerate}

\end{document}